\documentclass[journal]{IEEEtran}

\usepackage{setspace}

\usepackage{graphics,
           psfrag,
           epsfig,
           amsthm,
           cite,
           amssymb,
           url,
           dsfont,
           subfigure,
           algorithm,
           algorithmic,
           balance,
           enumerate,
           color,
           setspace
}
\usepackage{amsmath}%[centertags][tbtags][sumlimits][nosumlimits][intlimits][namelimits][nonamelimits]

\usepackage{epstopdf}

\newtheorem{lemma}{Lemma}

\newtheorem{theorem}{Theorem}

\newtheorem{proposition}{Proposition}

\newcommand{\eref}[1]{(\ref{#1})}
\newcommand{\sref}[1]{Section~\ref{#1}}
\newcommand{\appref}[1]{Appendix~\ref{#1}}
\newcommand{\fref}[1]{Figure~\ref{#1}}

\newcommand{\cref}[1]{Constraint~\ref{#1}}
\newcommand{\thref}[1]{Theorem~\ref{#1}}

\newcommand{\lref}[1]{Lemma~\ref{#1}}

\newcommand{\ignore}[1]{}

\IEEEoverridecommandlockouts
\begin{document}

\title{\vspace{-.5cm}Hybrid Scheduling/Signal-Level Coordination in the Downlink of Multi-Cloud Radio-Access Networks}

\author{
   \authorblockN{Ahmed Douik$^{\dagger}$, Hayssam Dahrouj$^\dagger$, Tareq Y. Al-Naffouri$^{\dagger\ast}$, and Mohamed-Slim Alouini$^\dagger$\\}%
   \authorblockA{$^\dagger$King Abdullah University of Science and Technology (KAUST), Kingdom of Saudi Arabia \\
    $^\ast$King Fahd University of Petroleum and Minerals (KFUPM), Kingdom of Saudi Arabia \\
    Email: $^\dagger$\{ahmed.douik,hayssam.dahrouj,tareq.alnaffouri,slim.alouini\}@kaust.edu.sa \\
    $^\ast$naffouri@kfupm.edu.sa
    }
\vspace{-.8cm}
}

\maketitle
 
\begin{abstract}
In the context of resource allocation in cloud-radio access networks, recent studies assume either signal-level or scheduling-level coordination. This paper, instead, considers a hybrid level of coordination for the scheduling problem in the downlink of a multi-cloud radio-access network, as a means to benefit from both scheduling policies. Consider a multi-cloud radio access network, where each cloud is connected to several base-stations (BSs) via high capacity links, and therefore allows joint signal processing between them. Across the multiple clouds, however, only scheduling-level coordination is permitted, as it requires a lower level of backhaul communication. The frame structure of every BS is composed of various time/frequency blocks, called power-zones (PZs), and kept at fixed power level. The paper addresses the problem of maximizing a network-wide utility by associating users to clouds and scheduling them to the PZs, under the practical constraints that each user is scheduled, at most, to a single cloud, but possibly to many BSs within the cloud, and can be served by one or more distinct PZs within the BSs' frame. The paper solves the problem using graph theory techniques by constructing the conflict graph. The scheduling problem is, then, shown to be equivalent to a maximum-weight independent set problem in the constructed graph, in which each vertex symbolizes an association of cloud, user, BS and PZ, with a weight representing the utility of that association. Simulation results suggest that the proposed hybrid scheduling strategy provides appreciable gain as compared to the scheduling-level coordinated networks, with a negligible degradation to signal-level coordination.
\end{abstract}


\section{Introduction}\label{sec:int}

Next generation mobile radio systems ($5$G) are expected to undergo major architectural changes, so as to support the deluge in demand for mobile data services and increase capacity, energy efficiency and latency reduction \cite{6824752,6736746}. One way to boost throughput and coverage in dense data networks is by moving from the single high-powered base-station (BS) to the massive deployment of overlaying BSs of different sizes. Such architecture, however, is subject to high inter BS interference, especially with the full spectrum reuse strategy set by $5$G. Traditionally, interference mitigation is performed by coordinating the different BSs through massive signalling and message exchange. Such coordination technique, however, in addition to being energy inefficient \cite{6983623}, may not always be feasible given the capacity limits of the backhaul links.

A promising interference mitigation technique is the coordinated multi-point (CoMP) transmission \cite{5594708} that is obtained by connecting the different BSs to a central unit, known as the cloud, to form the so-called cloud radio access network (CRAN). Such configuration moves most of the fundamental network functionalities to the cloud side, thereby allowing a separation between the control plane and the data plane. The virtualization in CRANs provides the potential for efficient resources utilisation, joint BSs operation (joint transmission, encoding and decoding), and effective energy control.

Different levels of coordination in CRANs are studied in the past literature, namely the signal-level \cite{6799231,6588350,6786060} coordination and the scheduling-level \cite{6525475,6811617,117665} one. In signal-level coordinated CRANs \cite{6799231,6588350,6786060}, all the data streams of different users are shared among the different BSs allowing joint operations. However, such level of coordination necessitates considerable backhaul communication. On the other extreme, in scheduling-level coordinated CRANs \cite{6525475,6811617,117665}, the cloud is responsible only for the efficient allocation of the resource blocks of each BS which clearly requires much less backhaul communication. While more practical to implement, preventing joint signal processing in scheduling-level coordination limits the system performances. To benefit from the advantages of both scheduling policies, this paper proposes a hybrid scheduling scheme.

Consider the downlink of a multi-CRAN, where each cloud is connected to several BSs. The frame structure of every BS is composed of various time/frequency blocks, called power-zones (PZs), kept at fixed power level. This paper proposes a hybrid level of coordination for the scheduling problem. For BSs connected to the same cloud, associating users to PZs is performed assuming signal-level coordination. Across the multiple clouds, only scheduling-level coordination is permitted, as it requires a lower level of backhaul communication.

In this paper context, the hybrid scheduling problem denotes determining the strategy of assigning users to clouds across the network, under the system limitation that each user is scheduled at most to a single cloud otherwise inter-cloud signal-level coordination is required. However, across the connected BSs in one cloud, users can be connected to multiple BSs and different PZs within each transmit frame. Moreover, each PZ is scheduled to exactly one user. The paper is related in part to the classical works on scheduling, and in part to the multi-CRANs. In the classical literature of cellular systems, scheduling is often performed assuming a prior assignment of users to BSs, for example the popular proportionally fair scheduling investigated in \cite{6525475,5464705}. In CRANs, recent works on coordinated scheduling assume a single cloud processing, for example \cite{6811617,117665}. This paper is further related to the multi-cloud setup studied in \cite{6799231}, which, however, assumes a known users to clouds association.

The paper considers the coordinated scheduling in multi-CRAN with an objective of maximizing a generic utility function. The paper's main contribution is to solve the problem optimally using techniques inherited from graph theory, by constructing the conflict graph in which each vertex represents an association of cloud, user, BS and PZ. It reformulates the problem as a maximum-weight independent set problem that can be solved using efficient algorithms \cite{15522856,2155446,6848102,5341909}. The paper further considers each of the scheduling policy, i.e., either scheduling-level or signal-level coordination. It shows that each of these scheduling policies can be obtained using similar techniques. Simulation results suggest that the proposed hybrid scheduling strategy provides appreciable gain as compared to the scheduling-level coordinated networks, with a negligible degradation to signal-level coordination.

The rest of this paper is organized as follows: In \sref{sec:sys}, the system model and the problem formulation are presented. \sref{sec:mul} proposes a solution to the hybrid scheduling problem. \sref{sec:ful} presents the scheduling solution of signal and scheduling level coordinated networks. Simulation results are discussed in \sref{sec:sim} before concluding the paper in \sref{sec:con}.

\section{System Model and Problem Formulation}\label{sec:sys}

\subsection{System Model and Parameters}

Consider the downlink of a multi-CRAN of $C$ clouds serving $U$ users in total. The $C$ clouds are connected to a central cloud. Each cloud (except the central one), is connected to $B$ BSs and is responsible for the signal-level coordination of the connected BSs. \fref{fig:network1} illustrates a multi-CRAN formed by $U=21$ users, and $C=3$ clouds each coordinating $B=3$ BSs. Let $\mathcal{C}$ be the set of clouds in the system each coordinating the set of BSs $\mathcal{B}$. Let $\mathcal{U}$ be the set of users in the network ($|\mathcal{U}|= U$, where the notation $|\mathcal{X}|$ refers to the cardinality of a set $\mathcal{X}$). The transmit frame of each BS is composed of several time/frequency resource blocks maintained at fixed transmit power. In this paper, the generic term PZ is used to refer to a time/frequency resource block of a BS. Let $\mathcal{Z}$ be the set of the $Z$ PZs of the frame of one BS. The transmit power of the $z$th PZ in the $b$th BS of the $c$th cloud is fixed to $P_{cbz}$, $\forall \ (c,b,z) \in \mathcal{C} \times \mathcal{B} \times \mathcal{Z}$, where the notation $\mathcal{X} \times \mathcal{Y}$ refers to the Cartesian product of the two sets $\mathcal{X}$ and $\mathcal{Y}$. \fref{fig:frame} depicts the coordinated frames of the connected BSs in the $c$th cloud. This paper focuses on the scheduling optimization step for a fixed transmit paper. Optimization with respect to the power values $P_{cbz}$ is left for future research.

\begin{figure}[t]
\centering
\includegraphics[width=0.9\linewidth]{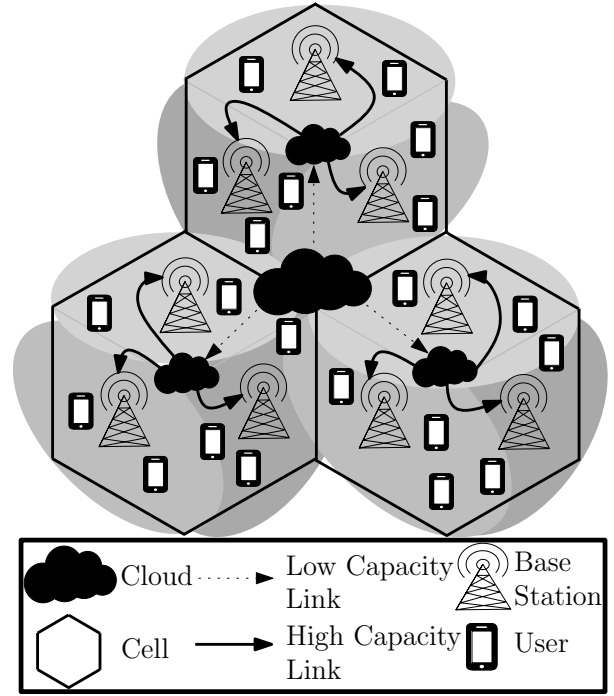}
\caption{Cloud enabled network composed $3$ cells, each containing $3$ base stations and $7$ users.}\label{fig:network1}
\end{figure}
\begin{figure}[t]
\centering
\includegraphics[width=0.9\linewidth]{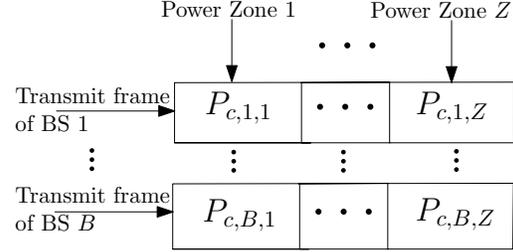}
\caption{Frame structure of $B$ base stations each containing $Z$ power zones.}\label{fig:frame}
\end{figure}

Each cloud $c \in \mathcal{C}$ is responsible for coordinating the different $B$ BSs allowing joint signal processing. The central cloud connecting all the clouds $c \in \mathcal{C}$, responsible for the scheduling policy, guarantees that the transmission of the different frames are synchronized across all BSs in the network ($CB$ BSs). Let $h_{cbz}^{u} \in \mathds{C},\ \forall \ (c,u,b,z) \in \mathcal{C} \times \mathcal{U} \times \mathcal{B} \times \mathcal{Z}$ be the complex channel gain from the $b$th BS of the $c$th cloud to user $u$ scheduled to PZ $z$. The signal-to-interference plus noise-ratio (SINR) of user $u$ when scheduled to PZ $z$ in the $b$th BS of the $c$th cloud can be expressed as:
\begin{align}
\text{SINR}_{cbz}^{u} = \cfrac{P_{cbz} |h_{cbz}^{u}|^2}{\Gamma(\sigma^2+ \sum\limits_{(c^\prime,b^\prime) \neq (c,b)} P_{c^\prime b^\prime z}|h_{c^\prime b^\prime z}^{u}|^2)},
\end{align}
where $\Gamma(.)$ denotes the SINR gap, and $\sigma^2$ is the Gaussian noise variance.

\subsection{Scheduling Problem Formulation}

The scheduling problem under investigation in this paper consists of assigning users to clouds and scheduling them to PZs in each BS frame under the following practical constraints.
\begin{itemize}
\item C1: Each user can connect at most to one cloud but possibly to many BSs in that cloud.
\item C2: Each PZ should be allocated to exactly one user.
\item C3: Each user cannot be served by the same PZ across different BSs.
\end{itemize}

Let $\pi_{cubz}$ be a generic network-wide benefit of assigning user $u$ to the $z$th PZ of the $b$th BS in the $c$th cloud. Let $X_{cubz}$ be a binary variable that is $1$ if user $u$ is mapped to the $z$th PZ of the $b$th BS in the $c$th cloud, and zero otherwise. Similarly, let $Y_{uz}$ be a binary variable that is $1$ if user $u$ is mapped to the $z$th PZ of any BS across the network, and zero otherwise. Further, let $Z_{cu}$ be a binary variable that is $1$ if user $u$ is assigned to cloud $c$. The scheduling problem this paper addresses can be formulated as the following 0-1 mixed integer programming problem:
\begin{subequations}
\label{Original_optimization_problem}
\begin{align}
\max & \sum_{c,u,b,z}\pi_{cubz}X_{cubz} \\
\label{constraint1}{\rm s.t.\ } &Z_{cu} = 1 - \delta \bigg(\sum_{b,z}X_{cubz}\bigg), \forall \ (c,u)\in\mathcal{C} \times \mathcal{U},\\
\label{constraint2}& \sum_{c}Z_{cu} \leq 1 , \quad \forall \ u \in\mathcal{U},\\
\label{constraint3}& \sum_{u}X_{cubz}=1, \quad\forall \ (c,b,z)\in \mathcal{C} \times \mathcal{B} \times \mathcal{Z},\\
\label{constraint4}& Y_{uz} = \sum_{cb} X_{cubz} \leq 1,\quad \forall \ (u,z) \in \mathcal{U} \times \mathcal{Z}, \\
\label{constraint5}& X_{cubz},Y_{uz},Z_{cu} \in \{0,1\},
\end{align}
\end{subequations}
where the optimization is over the binary variables $X_{cubz}$, $Y_{uz}$, and $Z_{cu}$ and the notation $\delta(.)$ refers to the discrete Dirac function which is equal to $1$ if its argument is equal to $0$ and $0$ otherwise. Both the equality constraint \eref{constraint1} and the inequality constraint \eref{constraint2} are due to system constraint C1. The equality constraints \eref{constraint3} and \eref{constraint4} correspond to the system constraints C2 and C3, respectively.

Using a generic solver for 0-1 mixed integer programs may require a search over the whole feasible space of solutions, i.e., all possible assignments of users to clouds and PZs of the network BSs. The complexity of such method is prohibitive for any reasonably sized network. The next section, instead, presents a more efficient method to solve the problem by constructing the conflict graph in which each vertex represents an association between clouds, users, BSs, and PZs. The paper reformulates the 0-1 mixed integer programming problem \eref{Original_optimization_problem} as a maximum-weight independent set problem in the conflict graph, which global optimum can be reached using efficient techniques, e.g., \cite{15522856,2155446}.

\section{Multi-Cloud Coordinated Scheduling}\label{sec:mul}

This section presents the optimal solution to the optimization problem \eref{Original_optimization_problem} by introducing the conflict graph and reformulating the problem as a maximum-weight independent set problem. The corresponding solution is naturally centralized, and the computation must be carried at the central cloud connecting all the clouds $c \in \mathcal{C}$.

\subsection{Conflict Graph Construction}

Define $\mathcal{A} = \mathcal{C} \times \mathcal{U} \times \mathcal{B} \times \mathcal{Z}$ as the set of all associations between clouds, users, BSs, and PZs, i.e., each element $a \in \mathcal{A}$ represents the association of one user to a cloud and a PZ in one of the connected BSs frame. For each association $a=(c,u,b,z) \in \mathcal{A}$, let $\pi(a)$ be the benefit of such association defined as $\pi(a)=\pi_{cubz}$. Let $\varphi_c$ be the cloud association function that maps each element from the set $\mathcal{A}$ to the corresponding cloud in the set $\mathcal{C}$. In other words, for $a=(c,u,b,z) \in \mathcal{A}$, $\varphi_c(a)=c$. Likewise, let $\varphi_u$, $\varphi_b$, and $\varphi_z$ be the association functions mapping each element $a=(c,u,b,z) \in \mathcal{A}$ to the set of users $\mathcal{U}$ (i.e., $\varphi_u(a)=u$), the set of BSs $\mathcal{B}$ (i.e., $\varphi_b(a)=b$), and the set of PZs (i.e., $\varphi_z(a)=z$), respectively.

The \emph{conflict} graph $\mathcal{G}(\mathcal{V},\mathcal{E})$ is an undirected graph in which each vertex represents an association of cloud, user, BS and PZ. Each edge between vertices represents a conflict between the two corresponding associations. Therefore, the conflict graph can be constructed by generating a vertex $v \in \mathcal{V}$ for each association $a \in \mathcal{A}$. Vertices $v$ and $v^\prime$ are conflicting vertices, and thus connected by an edge in $\mathcal{E}$ if one of the following connectivity conditions (CC) is true:
\begin{itemize}
\item CC1: $\delta(\varphi_u(v) - \varphi_u(v^\prime))(1-\delta(\varphi_c(v) - \varphi_c(v^\prime))) = 1$.
\item CC2: $(\varphi_c(v),\varphi_b(v),\varphi_z(v)) = (\varphi_c(v^\prime),\varphi_b(v^\prime),\varphi_z(v^\prime))$.
\item CC3: $\delta(\varphi_u(v) - \varphi_u(v^\prime))\delta(\varphi_z(v) - \varphi_z(v^\prime)) = 1$.
\end{itemize}

The connectivity constraint CC1 corresponds to a violation of the system constraint C1 as it describes that two vertices are conflicting if the same user is scheduled to different clouds. The connectivity constraint CC2 partially illustrates the system constraint C2, as it implies that each PZ should be associated to at most one user (not exactly one user as stated in the original system constraint). With the additional constraint (see \thref{th1} below) about the size of the independent set, CC2 becomes equivalent to C2. Finally, the edge creation condition CC3 perfectly translates a violation of the system constraint C3.

\begin{figure}[t]
\centering
\includegraphics[width=\linewidth]{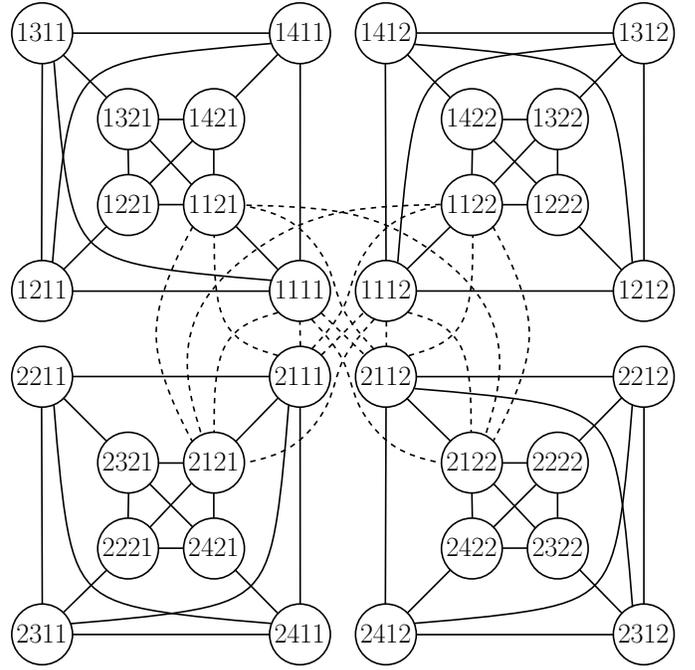}\\
\caption{Example of the conflict graph for a network composed of $2$ clouds, $2$ BSs per cloud, $2$ PZs per BS and a total of $4$ users. Intra-cloud connection are plotted in solid lines. Inter-cloud connections are illustrated only for user $1$ in dashed lines.}
\label{fig:graph}
\end{figure}

\fref{fig:graph} illustrates an example of the conflict graph in a multi-cloud system composed of $C=2$ clouds, $B=2$ BSs per cloud, $Z=2$ PZs per BS and $U=4$ users. Vertices, in this example, are labelled $cubz$, where $c$, $u$, $b$ and $z$ represent the indices of cloud, users, BSs, and PZs, respectively. In this example, there exist $4 \cdot 24 = 96$ independent sets of size $Z_{\text{tot}}=8$ that can be written in the following form:
\begin{enumerate}
\item \{$1a11$,\ $1a12$,\ $1b21$,\ $1b22$,\ $2c11$,\ $2c12$,\ $2d21$,\ $2z22$\}
\item \{$1a11$,\ $1a12$,\ $1b21$,\ $1b22$,\ $2c11$,\ $2d12$,\ $2d21$,\ $2y22$\}
\item \{$1a11$,\ $1b12$,\ $1b21$,\ $1a22$,\ $2c11$,\ $2c12$,\ $2d21$,\ $2z22$\}
\item \{$1a11$,\ $1b12$,\ $1b21$,\ $1a22$,\ $2c11$,\ $2d12$,\ $2d21$,\ $2y22$\},
\end{enumerate}
where $a,b,c,d \in \{1,2,3,4\}$ with $a \neq b \neq c \neq d$ ($24$ distinct permutations).

\subsection{Scheduling Solution}

Consider the conflict graph $\mathcal{G}(\mathcal{V},\mathcal{E})$ constructed above and let $\mathcal{I}$ be the set of all independent set of vertices of size $Z_{\text{tot}} = CBZ$. The following theorem characterises the solution of the optimization problem \eref{Original_optimization_problem}.

\begin{theorem}
The global optimal solution to the scheduling problem in multi-cloud network \eref{Original_optimization_problem} is the maximum-weight independent set among the independent sets of size $Z_{\text{tot}}$ in the conflict graph, where the weight of each vertex $v \in \mathcal{V}$ is given by:
\begin{align}
w(v) = \pi(v).
\end{align}
In other words, the optimal solution of the scheduling problem \eref{Original_optimization_problem} can be expressed as:
\begin{align}
I^* = \arg \max_{I \in \mathcal{I}} \sum_{v \in I} w(v).
\end{align}
\label{th1}
\end{theorem}
\begin{proof}
The proof can be found in \appref{app1}.
\end{proof}

In graph theory context, an independent set is a set in which each two vertices are not adjacent. The maximum-weight independent set problem is the problem of finding, in a weighted graph, the independent set(s) with the maximum weight where the weight of the set is defined as the sum of the individual weights of vertices belonging to the set. Maximum-weight independent set problems are well-known NP-hard problems. However, they can be solved efficiently, e.g., \cite{15522856,2155446}. Moreover, several approximate \cite{6848102} and polynomial time \cite{5341909} solving methods produce satisfactory results, in general.

\section{Extremes in Coordination Schemes}\label{sec:ful}

Different coordination levels are presented in this section. The fully coordinated, also known as the signal-level coordinated system, represents the optimal network design from a throughput point of view. However, such scheme requires a substantial amount of backhaul communication to share all the data streams between all BSs. The first part of this section illustrates the optimal scheduling in such coordination.

The second part of the section investigates the other scheduling-level extreme, wherein, user's data are processed in a single base-station. Such scheme, also known as scheduling-level coordination, has the merit of being cost effective since it requires only low capacity links to connect all BSs and clouds in the network. The second part of this section investigates the optimal user to BS and PZs assignment in such coordination. 

\subsection{Signal-Level Coordination}

For signal-level coordinated systems, all the data streams of users are shared among the BSs across the network. Hence, a user can be scheduled to many BSs in different clouds. The scheduling problem becomes the one of assigning users to clouds and scheduling them to PZs in each BS frame under the following practical constraints.
\begin{itemize}
\item Each PZ should be allocated to exactly one user.
\item Each user cannot be served by the same PZ across different BSs.
\end{itemize}

Following an analysis similar to the one in \sref{sec:mul}, the scheduling problem can be formulated as a 0-1 mixed integer programming as follows:
\begin{subequations}
\label{full_problem}
\begin{align}
\max& \sum_{c,u,b,z}\pi_{cubz}X_{cubz} \\
\label{constraint6} {\rm s.t.\ }& \sum_{u}X_{cubz}=1, \quad\forall \ (c,b,z)\in \mathcal{C} \times \mathcal{B} \times \mathcal{Z},\\
\label{constraint7}& \sum_{cb} X_{cubz} \leq 1,\quad \forall \ (u,z) \in \mathcal{U} \times \mathcal{Z}, \\
\label{constraint8}& X_{cubz} \in \{0,1\}, \forall \ (c,u,b,z) \in \mathcal{C} \times \mathcal{U}\times \mathcal{B}\times\mathcal{Z},
\end{align}
\end{subequations}
where the optimization is over the binary variable $X_{cubz}$, and where equations \eref{constraint6} and \eref{constraint7} correspond to the first and second system constraints, respectively. The following lemma provides the optimal solution to the optimization problem \eref{full_problem}.

\begin{lemma}
The optimal solution to the scheduling problem in signal-level coordinated cloud-enabled network \eref{full_problem} is the maximum-weight independent set of size $CBZ$ in the reduced conflict graph which is constructed in a similar manner as the conflict graph but using only connectivity constraint CC2 and CC3.
\label{l1}
\end{lemma}

\begin{proof}
The proof can be found in \appref{app4}.
\end{proof}

\subsection{Scheduling-Level Coordination}

In scheduling-level coordinated CRAN, the cloud is only responsible for scheduling users to BSs and PZs and synchronizing the transmit frames across the various BSs. In such coordinated systems, the scheduling problem is the one of assigning users to BSs and PZs under the following system constraints:
\begin{itemize}
\item Each user can connect at most to one BS but possibly to many PZs in that BS.
\item Each PZ should be allocated to exactly one user.
\end{itemize}

The scheduling problem can, then, be formulated as follows:
\begin{subequations}
\label{no_problem}
\begin{align}
\max & \sum_{c,u,b,z}\pi_{cubz}X_{cubz} \\
\label{constraint12}{\rm s.t.\ } &Y_{cub} = \min\bigg(\sum_{z}X_{cubz},1\bigg), \quad \forall \ (c,u,b) ,\\
\label{constraint11}& \sum_{c,b}Y_{cub}\leq 1,\quad \forall \ u\in\mathcal{U},\\
\label{constraint13}& \sum_{u}X_{cubz}=1, \quad \forall \ (c,b,z)\in \mathcal{C} \times \mathcal{B}\times\mathcal{Z},\\
\label{constraint14}& X_{cubz},Y_{cub} \in \{0,1\}, \quad \forall \ (c,u,b,z),
\end{align}
\end{subequations}
where the optimization is over the binary variables $X_{cubz}$ and $Y_{cub}$, where the constraints in \eref{constraint12} and \eref{constraint11} correspond to first system constraint, and where the equality constraint in \eref{constraint13} corresponds to the second system constraint.

Let the scheduling conflict graph $\mathcal{G}(\mathcal{V},\mathcal{E})$ be a constructed by generating a vertex $v \in \mathcal{V}$ for each association $a \in \mathcal{A}$. Vertices $v$ and $v^\prime$ are conflicting vertices, and thus connected by an edge in $\mathcal{E}$ if one of the following connectivity conditions is true:
\begin{itemize}
\item $\delta(\varphi_u(v) - \varphi_u(v^\prime))(1-\delta(\varphi_c(v) - \varphi_c(v^\prime))) = 1$.
\item $(\varphi_c(v),\varphi_b(v),\varphi_z(v)) = (\varphi_c(v^\prime),\varphi_b(v^\prime),\varphi_z(v^\prime))$.
\item $\delta(\varphi_u(v) - \varphi_u(v^\prime))(1-\delta(\varphi_b(v) - \varphi_b(v^\prime))) = 1$.
\end{itemize}

The following proposition characterizes the solution of the scheduling problem in scheduling-level coordinated CRANs:
\begin{proposition}
The optimal solution to the optimization problem \eref{no_problem} is the maximum-weight independent set of size $CBZ$ in the scheduling conflict graph.
\end{proposition}

\begin{proof}
The proof of this result is omitted as it follows the same steps as the proof of \thref{th1}.
\end{proof}

\section{Simulation Results}\label{sec:sim}

\begin{table}
\centering
\caption{System model parameters}
\begin{tabular}{|c|c|}
\hline
Cellular Layout & Hexagonal \\
\hline
Cell-to-Cell Distance & 500 meters \\
\hline
{Channel Model} & SUI-3 Terrain type B \\
\hline
Channel Estimation & Perfect \\
\hline
High Power & -42.60 dBm/Hz \\
\hline
Background Noise Power & -168.60 dBm/Hz \\
\hline
SINR Gap $\Gamma$ & 0dB\\
\hline
Bandwidth & 10 MHz \\
\hline
\end{tabular}
\label{t1}
\end{table}

The performance of the proposed scheduling schemes is shown in this section in the downlink of a cloud-radio access network, similar to \fref{fig:network1}. For illustration purposes, the simulations focus on the sum-rate maximization problem, i.e., $\pi_{cubz} = \log_2 (1 + \text{SINR}^u_{cbz})$. In these simulations, the cell size is set to $500$ meters and users are uniformly placed within each cell. The number of clouds, users, base-stations per cloud and power-zone per base-station frame change in each figure in order to quantify the gain in various scenarios. Simulations parameters are displayed in Table \ref{t1}.

\begin{figure}[t]
\centering
\includegraphics[width=1\linewidth]{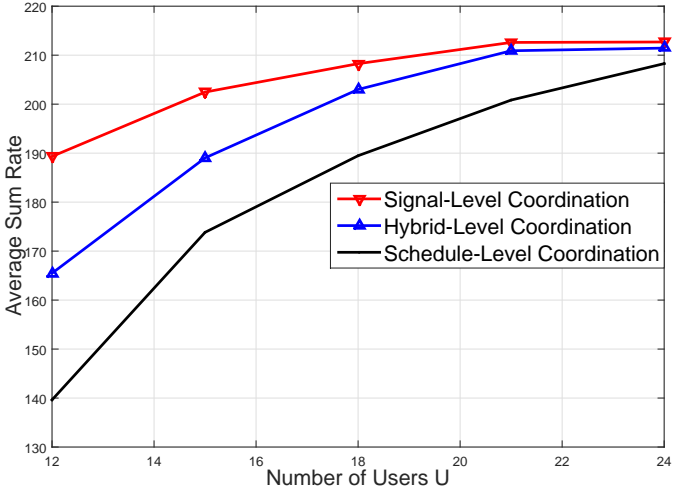}\\
\caption{Sum-rate in bps/Hz versus number of users $U$. Number of clouds is $C=3$ with $B=3$ base-stations per cloud, and $Z=5$ power-zones per BS's transmit frame.}\label{fig:U}
\end{figure}

\begin{figure}[t]
\centering
\includegraphics[width=1\linewidth]{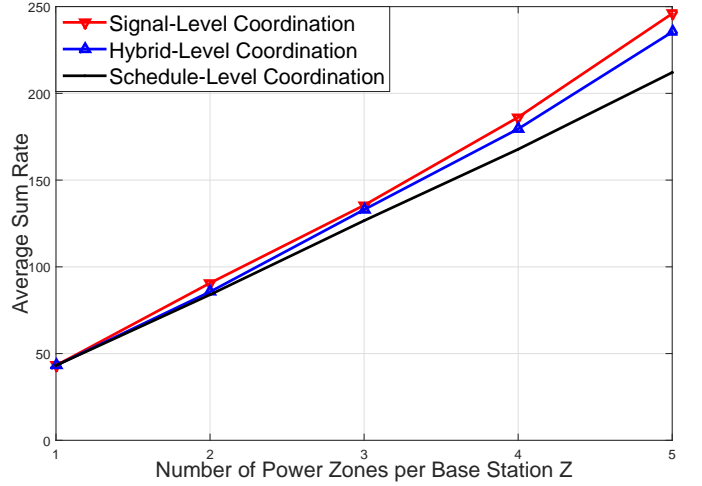}\\
\caption{Sum-rate in bps/Hz versus number of power-zones $Z$ per BS. Number of clouds is $C=3$ with $B=3$ base-stations per cloud, and $U=24$ users.}\label{fig:Z}
\end{figure}

\begin{figure}[t]
\centering
\includegraphics[width=1\linewidth]{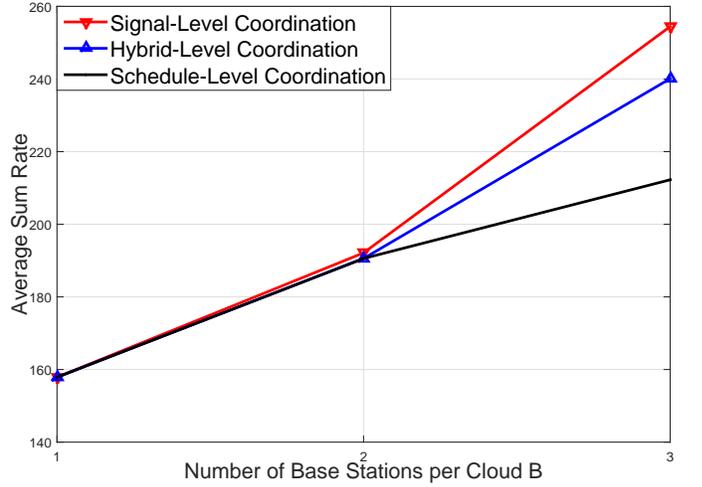}\\
\caption{Sum-rate in bps/Hz versus number of base-stations $B$ per cloud. Number of clouds is $C=3$ with $Z=5$ power-zones per BS's transmit frame, and $U=24$ users.}\label{fig:B}
\end{figure}

\begin{figure}[t]
\centering
\includegraphics[width=1\linewidth]{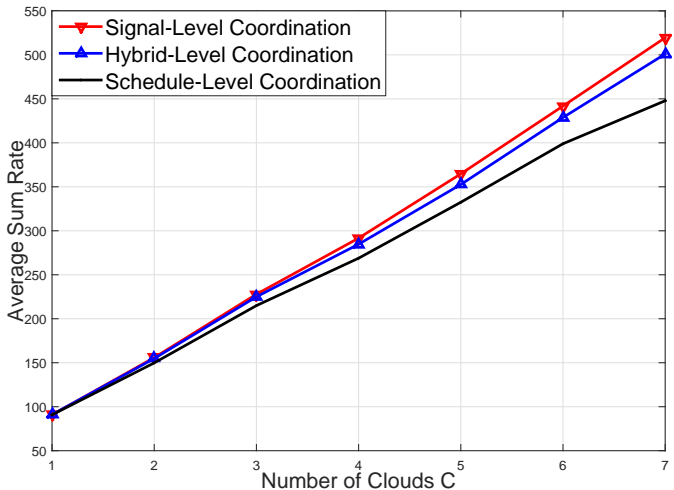}\\
\caption{Sum-rate in bps/Hz versus number of clouds $C$. Number of base-stations is $B=3$ per cloud, with $5$ power-zones per BS's transmit frame, and $U=8$ users per cloud.}\label{fig:C}
\end{figure}

\fref{fig:U} plots the sum-rate in bps/Hz versus the number of users $U$ for a CRAN composed of $C=3$ clouds, $B=3$ base-stations per cloud, and $Z=5$ power-zones per BS's transmit frame. The proposed hybrid coordination policy provides a significant gain against the scheduling-level coordinated system for a small number of users. As the number of users increases in the system, the different policies performs the same. This can be explained by the fact that as the number of users in the network increases, the probability that different users have the maximum pay-off in various PZs across the network increases which results in scheduling different users in different PZs and thus the different scheduling provide a similar performances.

\fref{fig:Z} plots the sum-rate in bps/Hz versus the number of power-zones $Z$ per BS for a network comprising $C=3$ clouds, $B=3$ base-stations, and $U=24$ users. From the system connectivity of the different policies, we clearly see that for a network comprising only one PZ per BS, the three scheduling are equivalent which explain the similar performance for $Z=1$. As the number of PZs per BS increases, the gap between the different coordinated systems increases. In fact, as the number of PZs increases, the ratio of users per PZ decreases and thus the role of the cloud as a scheduling entity becomes more pronounced.

\fref{fig:B} plots sum-rate in bps/Hz versus the number of base-stations $B$ per cloud for a network comprising $C=3$ clouds, $Z=5$ power-zones per BS's transmit frame, and $U=24$ users. For a small number of BSs, all the policies are equivalent and provide the same gain. However, as this number increases, the higher the level of coordination is, the more scheduling opportunities it offers. This explains the difference in performance as $B$ increases. We can see that our hybrid coordination provides a gain up to $13\%$ as compared to the scheduling-level coordinated network, for a degradation up to $6\%$ as compared to the signal-level coordination.

Finally, \fref{fig:C} plots the sum-rate in bps/Hz versus the number of clouds $C$ for a network comprising $B=3$ base-stations per cloud, $Z=5$ power-zones per BS's transmit frame, and $U=8$ users \emph{per cloud}. Again, our hybrid coordination provides a gain up to $12\%$ as compared with the scheduling-level coordination, for a negligible degradation up to $4\%$ against the signal-level coordinated system.

\section{Conclusion}\label{sec:con}

This paper considers the hybrid scheduling problem in the downlink of a multi-cloud radio-access network. The paper maximizes a network-wide utility under the practical constraint that each user is scheduled, at most, to a single cloud, but possibly to many BSs within the cloud and can be served by one or more distinct PZs within the BSs frame. The paper proposes a graph theoretical approach to solving the problem by introducing the conflict graph in which each vertex represents an association of cloud, user, BS and PZ. The problem is then reformulated as a maximum-weight independent set problem that can be efficiently solved. Finally, the paper shows that the optimal solution to the scheduling problem in different levels of system coordination can be obtained as a special case of the more general proposed system. Simulation results suggest that the proposed system architecture provides appreciable gain as compared to the scheduling-level coordinated networks, for a negligible degradation against the signal-level coordination.

\bibliographystyle{IEEEtran}
\bibliography{citations}

\begin{thebibliography}{10}
\providecommand{\url}[1]{#1}
\csname url@samestyle\endcsname
\providecommand{\newblock}{\relax}
\providecommand{\bibinfo}[2]{#2}
\providecommand{\BIBentrySTDinterwordspacing}{\spaceskip=0pt\relax}
\providecommand{\BIBentryALTinterwordstretchfactor}{4}
\providecommand{\BIBentryALTinterwordspacing}{\spaceskip=\fontdimen2\font plus
\BIBentryALTinterwordstretchfactor\fontdimen3\font minus
  \fontdimen4\font\relax}
\providecommand{\BIBforeignlanguage}[2]{{%
\expandafter\ifx\csname l@#1\endcsname\relax
\typeout{** WARNING: IEEEtran.bst: No hyphenation pattern has been}%
\typeout{** loaded for the language `#1'. Using the pattern for}%
\typeout{** the default language instead.}%
\else
\language=\csname l@#1\endcsname
\fi
#2}}
\providecommand{\BIBdecl}{\relax}
\BIBdecl

\bibitem{6824752}
J.~Andrews, S.~Buzzi, W.~Choi, S.~Hanly, A.~Lozano, A.~Soong, and J.~Zhang,
  ``What will 5{G} be?'' \emph{IEEE Journal on Selected Areas in
  Communications}, vol.~32, no.~6, pp. 1065--1082, June 2014.

\bibitem{6736746}
F.~Boccardi, R.~Heath, A.~Lozano, T.~Marzetta, and P.~Popovski, ``Five
  disruptive technology directions for 5{G},'' \emph{IEEE Communications
  Magazine}, vol.~52, no.~2, pp. 74--80, February 2014.

\bibitem{6983623}
M.~Peng, K.~Zhang, J.~Jiang, J.~Wang, and W.~Wang, ``Energy-efficient resource
  assignment and power allocation in heterogeneous cloud radio access
  networks,'' \emph{IEEE Transactions on Vehicular Technology}, vol.~PP,
  no.~99, pp. 1--1, 2014.

\bibitem{5594708}
D.~Gesbert, S.~Hanly, H.~Huang, S.~Shamai~Shitz, O.~Simeone, and W.~Yu,
  ``Multi-cell mimo cooperative networks: A new look at interference,''
  \emph{IEEE Journal on Selected Areas in Communications}, vol.~28, no.~9, pp.
  1380--1408, December 2010.

\bibitem{6799231}
S.-H. Park, O.~Simeone, O.~Sahin, and S.~Shamai, ``Inter-cluster design of
  precoding and fronthaul compression for cloud radio access networks,''
  \emph{IEEE Wireless Communications Letters}, vol.~3, no.~4, pp. 369--372, Aug
  2014.

\bibitem{6588350}
------, ``Joint precoding and multivariate backhaul compression for the
  downlink of cloud radio access networks,'' \emph{IEEE Transactions on Signal
  Processing}, vol.~61, no.~22, pp. 5646--5658, Nov 2013.

\bibitem{6786060}
Y.~Shi, J.~Zhang, and K.~Letaief, ``Group sparse beamforming for green
  cloud-ran,'' \emph{IEEE Transactions on Wireless Communications}, vol.~13,
  no.~5, pp. 2809--2823, May 2014.

\bibitem{6525475}
W.~Yu, T.~Kwon, and C.~Shin, ``Multicell coordination via joint scheduling,
  beamforming, and power spectrum adaptation,'' \emph{IEEE Transactions on
  Wireless Communications}, vol.~12, no.~7, pp. 1--14, July 2013.

\bibitem{6811617}
H.~Dahrouj, W.~Yu, T.~Tang, J.~Chow, and R.~Selea, ``Coordinated scheduling for
  wireless backhaul networks with soft frequency reuse,'' in \emph{Proc. of the
  21st Europea Signal Processing Conference (EUSIPCO' 2013), Marrakech,
  Morocco}, Sept 2013, pp. 1--5.

\bibitem{117665}
A.~Douik, H.~Dahrouj, T.~Y. Al-Naffouri, and M.-S. Alouini, ``Coordinated
  scheduling for the downlink of cloud radio-access networks,'' \emph{Proc. of
  {IEEE} International Conference on Communications (ICC' 2015), London, UK.
  avaialble Arxiv e-prints}, vol. abs/1411.4144, 2015.

\bibitem{5464705}
B.~Rengarajan, A.~Stolyar, and H.~Viswanathan, ``Self-organizing dynamic
  fractional frequency reuse on the uplink of ofdma systems,'' in \emph{Proc.
  of 2010 44th Annual Conference on Information Sciences and Systems (CISS'
  2010), Princeton, New Jersey, USA}, March 2010, pp. 1--6.

\bibitem{15522856}
F.~V. Fomin, F.~Grandoni, and D.~Kratsch, ``A measure \& conquer approach for
  the analysis of exact algorithms,'' \emph{Journal of the ACM}, vol.~56,
  no.~5, pp. 25:1--25:32, Aug. 2009.

\bibitem{2155446}
N.~Bourgeois, B.~Escoffier, V.~T. Paschos, and J.~M.~M. van Rooij, ``A
  bottom-up method and fast algorithms for max independent set,'' in
  \emph{Proc. of the 12th Scandinavian Conference on Algorithm Theory (SWAT'
  2010), Bergen, Norway}.

\bibitem{6848102}
P.~jun Wan, X.~Jia, G.~Dai, H.~Du, and O.~Frieder, ``Fast and simple
  approximation algorithms for maximum weighted independent set of links,'' in
  \emph{Proc. of 33th IEEE Conference on Computer Communications (INFOCOM'
  2009), Toronto, canada}, April 2014, pp. 1653--1661.

\bibitem{5341909}
N.~Esfahani, P.~Mazrooei, K.~Mahdaviani, and B.~Omoomi, ``A note on the p-time
  algorithms for solving the maximum independent set problem,'' in \emph{Proc.
  of 2nd Conference on Data Mining and Optimization (DMO' 2009), Bandar Baru
  Bangi, Malaysia}, Oct 2009, pp. 65--70.

\end{thebibliography}

\appendices
\numberwithin{equation}{section}

\section{Proof of \thref{th1}}\label{app1}

To prove the result, the optimization problem \eref{Original_optimization_problem} is first reformulated as a search over the set of feasible schedules. Further, a one to one mapping between the possible schedules and the set of independent sets of size $Z_{tot}$ in the conflict graph is highlighted. Showing that the weight of each independent set is the objective function of \eref{Original_optimization_problem} indicates that the optimal solution is the maximum-weight independent set which concludes the proof.

All possible schedules representing the assignments between clouds, users, BSs and PZs, regardless of the feasibility, can be conveniently represented by the set of all subsets of $\mathcal{A}$, i.e., the power set $\mathcal{P}(\mathcal{A})$ of the set of associations $\mathcal{A}$. Recall that for an association $a=(c,u,b,z)$ in a schedule $\mathcal{S} \subseteq \mathcal{A}$ (i.e., $\mathcal{S} \in \mathcal{P}(\mathcal{A})$), the benefit of the association is given by $\pi(a)=\pi_{cubz}$. The following lemma reformulates the multi-cloud joint scheduling problem.
\begin{lemma}
The discrete optimization problem \eref{Original_optimization_problem} can be written as follows:
\begin{align}
\label{new_optimization_problem}
&\max_{\mathcal{S} \in \mathcal{P}(\mathcal{A})} \sum_{a \in \mathcal{S}} \pi(a) \\
&{\rm s.t.\ } \mathcal{S} \in \mathcal{F},
\end{align}
where $\mathcal{F}$ is the set of feasible schedules defined as follows:
\begin{subequations}
\begin{align}
&\mathcal{F}= \{ \mathcal{S} \in \mathcal{P}(\mathcal{A}) \text{ such that } \forall \ a \neq a^\prime \in \nonumber \mathcal{S}\\
& \delta(\varphi_u(a) - \varphi_u(a^\prime)) (1-(\delta(\varphi_c(a)-\varphi_c(a^\prime))) = 0, \label{eq1} \\
& (\varphi_c(a),\varphi_b(a),\varphi_z(a)) \neq (\varphi_c(a^\prime),\varphi_b(a^\prime),\varphi_z(a^\prime)) \label{eq2}, \\
& \delta(\varphi_u(a) - \varphi_u(a^\prime))\delta(\varphi_z(a) - \varphi_z(a^\prime)) =0 \label{eq3} \\
&|\mathcal{S}| = Z_{\text{tot}} \label{eq4} \}.
\end{align}
\end{subequations}
\label{l2}
\end{lemma}

\begin{proof}
The proof can be found in \appref{app5}.
\end{proof}

To demonstrate that there is a one to one mapping between the set of feasible schedules $\mathcal{F}$ and the set of independent sets $\mathcal{I}$ of size $Z_{\text{tot}}$, we first show that each element of $\mathcal{F}$ is represented by a unique element in $\mathcal{I}$. We, then, show that each independent set can uniquely be represented by a feasible schedule.

Let the feasible schedule $\mathcal{S} \in \mathcal{F}$ be associated with the set of vertices $I$ in the conflict graph. Assume $\exists \ v \neq v^\prime \in I$ such that $v$ and $v^\prime$ are connected. From the connectivity conditions in the conflict graph, vertices $v$ and $v^\prime$ verify one of the following conditions
\begin{itemize}
\item CC1: $\delta(\varphi_u(v) - \varphi_u(v^\prime))(1-\delta(\varphi_c(v) - \varphi_c(v^\prime))) = 1$: this condition violate the constraint \eref{eq1} of the construction of $\mathcal{F}$.
\item CC2: $(\varphi_c(v),\varphi_b(v),\varphi_z(v)) = (\varphi_c(v^\prime),\varphi_b(v^\prime),\varphi_z(v^\prime))$: this condition violate the constraint \eref{eq2} of the construction of $\mathcal{F}$.
\item CC3: $\delta(\varphi_u(v) - \varphi_u(v^\prime))\delta(\varphi_z(v) - \varphi_z(v^\prime)) = 1$: this condition violate the constraint \eref{eq3} of the construction of $\mathcal{F}$.
\end{itemize}
Therefore, each pair of vertices $v \neq v^\prime \in I$ are not connected which demonstrate that $I$ is an independent set of vertices in the conflict graph. Finally, from the construction constraint \eref{eq4}, $\mathcal{S}$ and by extension $I$ have $Z_{\text{tot}}$ association. Therefore, $I$ is a set of $Z_{\text{tot}}$ independent vertices which concludes that $I \in \mathcal{I}$. The uniqueness of $I$ follows directly from the bijection between the set of vertices in the graph and the set of associations in $\mathcal{A}$.

To establish the converse, let $I \in \mathcal{I}$ be an independent set of size $Z_{\text{tot}}$ and let $\mathcal{S}$ be its corresponding schedule. Using an argument similar to the one in previous paragraph, it can be easily shown that all the associations in $\mathcal{S}$ verify the constraints \eref{eq1}, \eref{eq2}, and \eref{eq3}. Given that $I$ is of size $Z_{\text{tot}}$, then $\mathcal{S}$ verify \eref{eq4} which concludes that $\mathcal{S} \in \mathcal{F}$. Uniqueness of the element is given by the same argument as earlier.

To conclude the proof, note that the weight of an independent set $I \in \mathcal{I}$ is equal to the objective function \eref{new_optimization_problem} and by extension to the original objective function \eref{Original_optimization_problem}. Therefore The global optimal solution of the joint scheduling problem in multi-cloud network \eref{Original_optimization_problem} is equivalent to a maximum-weight independent set among the independent sets of size $Z_{\text{tot}}$ in the conflict graph.

\section{Proof of \lref{l1}}\label{app4}

Note that the constraints \eref{constraint6}, \eref{constraint7} and \eref{constraint8} of the optimization problem \eref{full_problem} are the same constraint as \eref{constraint3}, \eref{constraint4} and \eref{constraint5}, respectively, in the original optimization problem \eref{Original_optimization_problem}. Therefore, this lemma can be proved using steps similar to the one used in \thref{th1}.

Let $\mathcal{F} \subset \mathcal{P}(\mathcal{A})$ be the set of feasible schedules. Given the mapping between the original constraints of the problem and the constraints of constructing the set $\mathcal{F}$ illustrated in \lref{l2}, it can be easily shown that problem \eref{full_problem} can be written as follows:
\begin{align}
&\max_{\mathcal{S} \in \mathcal{P}(\mathcal{A})} \sum_{a \in \mathcal{S}} \pi(a) \\
&{\rm s.t.\ } \mathcal{S} \in \mathcal{F},
\end{align}
where $\mathcal{F}$ is the set of feasible schedules defined as follows:
\begin{subequations}
\begin{align}
&\mathcal{F}= \{ \mathcal{S} \in \mathcal{P}(\mathcal{A}) \text{ such that } \forall \ a \neq a^\prime \in \nonumber \mathcal{S}\\
& (\varphi_c(a),\varphi_b(a),\varphi_z(a)) \neq (\varphi_c(a^\prime),\varphi_b(a^\prime),\varphi_z(a^\prime)) , \\
& \delta(\varphi_u(a) - \varphi_u(a^\prime))\delta(\varphi_z(a) - \varphi_z(a^\prime)) =0 \\
&|\mathcal{S}| = Z_{\text{tot}} \}.
\end{align}
\end{subequations}

Let the reduced conflict graph be constructed by generating a vertex of each association $a \in \mathcal{A}$ and connecting two distinct vertices $v$ and $v^\prime$ if one of the following two conditions holds:
\begin{itemize}
\item CC2: $(\varphi_c(v),\varphi_b(v),\varphi_z(v)) = (\varphi_c(v^\prime),\varphi_b(v^\prime),\varphi_z(v^\prime))$.
\item CC3: $\delta(\varphi_u(v) - \varphi_u(v^\prime))\delta(\varphi_z(v) - \varphi_z(v^\prime)) = 1$.
\end{itemize}

Define $\mathcal{I}$ as the set of the independent set of vertices of size $Z_{\text{tot}}$ in the reduced conflict graph. Following steps similar to the one used in \thref{th1}, it can be shown that there is a one to one mapping between the set of feasible schedule $\mathcal{F}$ and the set $\mathcal{I}$ and that the objective function is represented by the sum of the weight of the vertices in the independent set. As a conclusion, the optimal solution to the scheduling problem \eref{full_problem} in signal-level coordinated cloud-enabled network is the maximum-weight independent set of size $CBZ$ in the reduced conflict graph.

\section{Proof of \lref{l2}}\label{app5}

To prove this lemma, it is sufficient to prove to that the objective function and the constraints of \eref{Original_optimization_problem} are equivalent to those of the optimization problem \eref{new_optimization_problem}. The objective function of \eref{Original_optimization_problem} is equivalent to the one of \eref{new_optimization_problem} as shown in the following equation:
\begin{align}
\sum_{c,u,b,z}\pi_{cubz}X_{cubz} = \sum_{a \in \mathcal{A}}\pi(a)X(a) = \sum_{a \in \mathcal{S}}\pi(a),
\end{align}
where $X(a)$ is defined in the same manner as $\pi(a)$, i.e., $X(a)= X_{cubz}$ for $a=(c,u,b,z) \in \mathcal{A}$ and $\mathcal{S} =\{a \in \mathcal{A} \ | \ X(a)=1 \}$. Therefore, the two objective functions are equivalent:
\begin{align}
\max \sum_{c,u,b,z}\pi_{cubz}X_{cubz} = \max_{\mathcal{S} \in \mathcal{P}(\mathcal{A})} \sum_{a \in \mathcal{S}}\pi(a).
\end{align}

In what follows, the constraints \eref{constraint1} and \eref{constraint2} are shown to be equivalent to the constraint \eref{eq1}, the constraint \eref{constraint3} is proven to be equivalent to \eref{eq2} and \eref{eq4}. Finally to conclude the proof, \eref{constraint4} is demonstrated to be the same constraint as \eref{eq3}.

Define $\mathcal{S}_{cu} \subset \mathcal{S}$ as the set of associations in schedule $\mathcal{S}$ concerning the $c$th cloud and the $u$th user. The expression of the set is the following:
\begin{align}
\mathcal{S}_{cu} = \left\{ a \in \mathcal{S} \ | \ \varphi_c(a)=c ,\ \varphi_u(a)=u \right\}.
\end{align}
Let $\mathcal{S}_{u} \subset \mathcal{P}(\mathcal{S})$ be the set of all the set concerning user $u$ defined as:
\begin{align}
\mathcal{S}_{u} = \left\{ \mathcal{S}_{cu}, c \in \mathcal{C} \right\}.
\end{align}
The constraints \eref{constraint1} (i.e., $Z_{cu} = 1 - \delta \bigg(\sum_{b,z}X_{cubz}\bigg)$) and \eref{constraint2} (i.e., $\sum_{c}Z_{cu} \leq 1$) are equivalent to the following constraint
\begin{align}
Z_{cu} = 1 - \delta \bigg(\sum_{b,z}X_{cubz}\bigg) \leq 1 \Leftrightarrow |\mathcal{S}_u| \leq 1 , \ \forall \ u.
\label{l2:eq6}
\end{align}
We now show that the inequality $|\mathcal{S}_u| \leq 1$ is equivalent to the following equality $\forall \ a \neq a^\prime \in \mathcal{S}$:
\begin{align}
\delta(\varphi_u(a) - \varphi_u(a^\prime)) (1-(\delta(\varphi_c(a)-\varphi_c(a^\prime))) = 0
\label{l2:eq7}
\end{align}

First note that if $a \in \mathcal{S}_{u}$ and $a^\prime \in \mathcal{S}_{u^\prime}$ with $u \neq u^\prime$, then $\varphi_u(a) \neq \varphi_u(a^\prime)$ which concludes that \eref{l2:eq7} holds for such $a$ and $a^\prime$. Now let $a \neq a^\prime \in \mathcal{S}_{u}$. Since $|\mathcal{S}_u| \leq 1$ then $\exists$ unique $c \in \mathcal{C}$ such that $\mathcal{S}_{cu} \neq \varnothing$. Hence $a \neq a^\prime \in \mathcal{S}_{cu}$, i.e., $\varphi_c(a) = \varphi_c(a^\prime)$ which concludes that \eref{l2:eq7} holds for such $a$ and $a^\prime$. Given that $\mathcal{S}$ can be written as $\bigcup_{u} \mathcal{S}_u$, then \eref{l2:eq7} is valid $\forall \ a \neq a^\prime \in \mathcal{S}$. Combining \eref{l2:eq6} and \eref{l2:eq7} proves that the constraints \eref{constraint1} and \eref{constraint2} are equivalent to the constraint \eref{eq1}.

Define $\mathcal{S}_{cbz} \subset \mathcal{S}$ as the set of associations in schedule $\mathcal{S}$ concerning the $z$th PZ in the $b$th BS connected to the $c$th cloud. The expression of the set is the following:
\begin{align}
\mathcal{S}_{cbz} = \left\{ a \in \mathcal{S} \ | \ \varphi_c(a)=c ,\ \varphi_b(a)=b ,\ \varphi_z(a)=z \right\}.
\end{align}
The constraint \eref{constraint3} can be written as a function of the partial schedules as follows:
\begin{align}
\sum_{u}X_{cubz}=1 \Leftrightarrow |\mathcal{S}_{cbz}|=1, \ \forall \ (c,b,z).
\label{l2:eq1}
\end{align}
Assume $\exists \ a \neq a^\prime \in \mathcal{S}$ such that $\varphi_c(a)=\varphi_c(a^\prime) ,\ \varphi_b(a)=\varphi_b(a^\prime)$ , and $\varphi_z(a)=\varphi_z(a^\prime)$. It is clear that $a,a^\prime \in \mathcal{S}_{cbz}$ where $c=\varphi_c(a) ,\ b=\varphi_b(a)$, and $z=\varphi_z(a)$. However, from \eref{l2:eq1}, we have $|\mathcal{S}_{cbz}|=1$. Therefore, $a = a^\prime$ which concludes that, $\forall \ a \neq a^\prime \in \mathcal{S}$, we have:
\begin{align}
(\varphi_c(a),\varphi_b(a),\varphi_z(a)) \neq (\varphi_c(a^\prime),\varphi_b(a^\prime),\varphi_z(a^\prime)).
\label{l2:eq2}
\end{align}
We now show that $\mathcal{S}_{cbz} \cap \mathcal{S}_{c^\prime b^\prime z^\prime} = \varnothing$ for all sets in which at least one of the following holds: $c \neq c^\prime$, and/or $b \neq b^\prime$, and/or $z \neq z^\prime$. From \eref{l2:eq1}, both sets contains a single association, hence $\mathcal{S}_{cbz} \cap \mathcal{S}_{c^\prime b^\prime z^\prime} \neq \varnothing$ means that $\mathcal{S}_{cbz} = \mathcal{S}_{c^\prime b^\prime z^\prime}$ which do not hold since $c \neq c^\prime$, and/or $b \neq b^\prime$, and/or $z \neq z^\prime$. As a conclusion, the cardinality of the schedule $\mathcal{S}$ can be written as:
\begin{align}
|\mathcal{S}| = \left|\bigcup_{c,b,z} \mathcal{S}_{cbz} \right| = \bigcup_{c,b,z} \left|\mathcal{S}_{cbz} \right| = CBZ = Z_{\text{tot}} .
\label{l2:eq3}
\end{align}
The combination of equations \eref{l2:eq1}, \eref{l2:eq2} and \eref{l2:eq3} shows that the constraint \eref{constraint3} is equivalent to \eref{eq2} and \eref{eq4}.

Define $\mathcal{S}_{uz} \subset \mathcal{S}$ as the set of associations in schedule $\mathcal{S}$ concerning the $u$th user scheduled in the $z$th PZ of one of the connected BS. The expression of the set is the following:
\begin{align}
\mathcal{S}_{uz} = \left\{ a \in \mathcal{S} \ | \ \varphi_u(a)=u ,\ \varphi_z(a)=z \right\}.
\end{align}
The constraint \eref{constraint4} can be written as a function of the partial schedules as follows:
\begin{align}
Y_{uz} = \sum_{cb} X_{cubz} \leq 1 \Leftrightarrow |\mathcal{S}_{uz}| \leq 1, \ \forall \ (u,z).
\label{l2:eq4}
\end{align}
To conclude the proof, it is sufficient to show that, if $ |\mathcal{S}_{uz}| \leq 1, \ \forall \ (u,z)$, then the following equation holds for $\ a \neq a^\prime \in \mathcal{S}$:
\begin{align}
\delta(\varphi_u(a) - \varphi_u(a^\prime))\delta(\varphi_z(a) - \varphi_z(a^\prime)) =0.
\label{l2:eq5}
\end{align}
Let the schedule be partitioned into partial schedules as follows $\mathcal{S} = \bigcup_{uz}\mathcal{S}_{uz}$. For $a \in \mathcal{S}_{uz}$ and $a^\prime \in \mathcal{S}_{u^\prime z^\prime} \neq \mathcal{S}_{uz}$, it is clear that either $u \neq u^\prime$ and/or $z \neq z^\prime$. Hence, equality \eref{l2:eq5} holds for all $a \in \mathcal{S}_{uz}$ and $a^\prime \in \mathcal{S}_{u^\prime z^\prime} \neq \mathcal{S}_{uz}$. Given that $|\mathcal{S}_{uz}| \leq 1$, then $\nexists \ a \neq a^\prime \in \mathcal{S}_{uz}, \ \forall \ (u,z)$ which concludes that \eref{l2:eq5} is verified. The combination of equations \eref{l2:eq4}, and \eref{l2:eq5} shows that the constraint \eref{constraint4} is equivalent to \eref{eq3}.

\end{document}